\documentclass[graybox]{svmult}
\usepackage{mathptmx}       
\usepackage{helvet}         
\usepackage{courier}        
\usepackage{type1cm}        
\usepackage{makeidx}         
\usepackage{graphicx}        
\usepackage{multicol}        
\usepackage[bottom]{footmisc}
\usepackage{amsmath,xcolor}
\usepackage{epstopdf}
\usepackage{url}
\usepackage{tikz}
\usepackage{amssymb}
\usepackage{caption}
\usepackage{subfigure}

\newtheorem{assumption}{Assumption}
\pdfoutput=1
\makeindex             

\begin{document}

\title*{$H$-infinity Filtering for Cloud-Aided Semi-active Suspension with Delayed Information}
\author{Zhaojian Li and Ilya Kolmanovsky and Ella Atkins and Jianbo Lu and Dimitar Filev}
\institute{Zhaojian Li \at Department of Aerospace Engineering, University of Michigan, 1221 Beal Avenue, Ann Arbor, MI, USA 48109, \email{zhaojli@umich.edu}
\and Ilya Kolmanovsky \at Department of Aerospace Engineering, University of Michigan, 1221 Beal Avenue, Ann Arbor, MI, USA 48109, \email{ilya@umich.edu}
\and Ella Atkins \at Department of Aerospace Engineering, University of Michigan, 1221 Beal Avenue, Ann Arbor, MI, USA 48109, \email{ematkins@umich.edu}
\and Jianbo Lu \at Research \& Advanced Engineering, Ford Motor Company, 2101 Village Rd, Dearborn, MI, USA 48124, \email{jlu10@ford.com}
\and Dimitar Filev \at Research \& Advanced Engineering, Ford Motor Company, 2101 Village Rd, Dearborn, MI, USA 48124, \email{dfilev@ford.com}}
%
%

\maketitle

\abstract{This chapter presents an $H_{\infty}$ filtering framework for cloud-aided semi-active suspension system with time-varying delays. In this system, road profile information is downloaded from a cloud database to facilitate onboard estimation of suspension states. Time-varying data transmission delays are considered and assumed to be bounded. A quarter-car linear suspension model is used and an $H_{\infty}$ filter is designed with both onboard sensor measurements and delayed road profile information from the cloud. The filter design procedure is designed based on linear matrix inequalities (LMIs). Numerical simulation results are reported that illustrates the fusion of cloud-based and on-board information that can be achieved in Vehicle-to-Cloud-to-Vehicle (V2C2V) implementation.}

\section{Introduction}

The interest in employing cloud computing for automotive applications is growing to support computation and data intensive tasks \cite{jianbo,mobile,car_cloud}. The cloud can provide access to ``big data'' as well as real-time crowd-sourced information. Smart utilization of on-demand cloud resources can increase situation awareness and provide additional functionality. In addition, computation and data intensive tasks can be out-sourced to the cloud, enabling advanced and computation intensive algorithms to be implemented in real time. While embedded vehicle processors remain essential for time-critical applications, cloud computing can extend current control functionalities with additional functions to enhance performance \cite{jianbo}.

Numerous automotive functions have been identified as candidates for Vehicle-to-Cloud-to-Vehicle (V2C2V) implementations \cite{mobile}. In particular, a Cloud-aided safety-based route planning system has been proposed that exploits road risk index database and real-time factors like traffic and weather, and generates a ``safe'' route \cite{safety,safetyA}. A cloud-based road anomaly detection and crowd-sourcing framework has been proposed in \cite{pothole}. The cloud-aided vehicle semi-active suspension control system is another potential application \cite{suspension}, in which road profile data from the cloud is exploited. In this chapter, we consider the state estimation problem for the cloud-aided semi-active suspension system.



While cloud database can provide large quantity of information, its use for vehicle control is hindered by inevitable time delays in information transmission. While there are proposed mechanisms (e.g., exploiting multiple communication channels \cite{mult1,mult2}) to alleviate the effects of time delays, it is well acknowledged that delays can cause system instability and performance degradation and thus have to be considered in control and filter design. The architecture of V2C2V system with time delays is illustrated in  Figure.~\ref{info-based}. When needed, the vehicle can send a data request together with its GPS coordinates to the cloud. Then the cloud will send the requested data to the vehicle. The messages are exchanged with a wireless communication channel where vehicle-to-cloud delay ($\tau_{v2c}$) and cloud-to-vehicle delay ($\tau_{c2v}$) occur.

\begin{figure}[!t]
\centering
\includegraphics[width=2.5in]{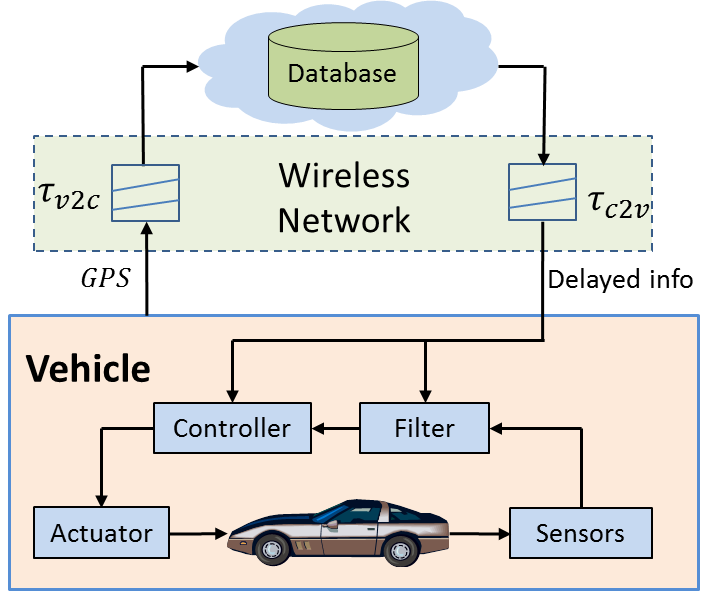}
\caption{Info-based V2C2V system with time delays.}\label{info-based}
\end{figure}

In this chapter, we consider the design of an observer for a V2C2V semi-active suspension system. Unlike traditional white-noise road disturbance treatment \cite{a,b,c}, in our system road profile information is downloaded from the cloud to facilitate the suspension state estimation. There are two main approaches in filter design for linear systems. One is $H_2$ filtering which minimizes estimation error variance. The other is $H_{\infty}$ filtering which is based on $H_{\infty}$ performance criterion. $H_{\infty}$ techniques are popular in stability analysis and filter design, even in cases with time delays \cite{H1,MCC}. In this chapter, we develop an $H_{\infty}$ filtering approach for V2C2V semi-active suspension with time-varying information delays. The filter design is reduced to linear matrix inequalities (LMIs), which can be solved by standard LMI solvers.

This book chapter is an extension of our previous conference paper \cite{Li}. The main contribution of this chapter is the illustration of the potential to fuse cloud-based and on-board measurements in V2C2V systems to facilitate state estimation and the developments of H-infinity filtering framework for handling communication delays.

 The rest of this chapter is organized as follows. In Section~\ref{sec:II} we present the preliminaries and problem formulation of the $H_{\infty}$ filtering for info-based V2C2V semi-active suspension system. In Section~\ref{sec:III}, stability and $H_\infty$ filter performance analysis results are presented and derived in terms of LMIs. Section~\ref{sec:IV} presents the design procedure for the H-infinity filter. Numerical simulations are developed in Section~\ref{sec:V}. Section~\ref{sec:VI} concludes the chapter.


\section{Problem formulation}\label{sec:II}
 In this chapter, we consider the filtering problem for cloud-aided semi-active suspension. Quarter-car models are often used for suspension control design \cite{a,b,c}, because they are simple yet capture many important characteristics of the full-car model. A quarter-car model, with 2 degrees of freedom (DOF), as shown in Fig.~\ref{fig:a}, is used. The $M_s$ and $M_{us}$ represent the car body (sprung) mass and the tire and axles (unsprung mass), respectively. The spring and shock absorber with adjustable damping ratio constitute the suspension system, connecting sprung (body) and unsprung (wheel assembly) masses. The tire is modeled as a spring with stiffness $k_{us}$ and its damping ratio is assumed to be negligible in the suspension formulation.
From Fig.~\ref{fig:a}, we have the following equations of motion:
\begin{equation}\label{equ:new1}
\begin{aligned}
\dot{x}_1&=x_2-\alpha w-\dot{r}_o,\\
M_{us} \dot{x}_2&=-k_{us} x_1+k_s x_3+c_s(x_4-x_2)+u,\\
\dot{x}_3&=x_4-x_2,\\
M_s \dot{x}_4&=-k_s x_3-c_s(x_4-x_2)-u,
\end{aligned}
\end{equation}
where $x_1$ is the tire deflection from equilibrium; $x_2$ is the unsprung mass velocity; $x_3$ is the suspension deflection from equilibrium; $x_4$ is the sprung mass velocity; $\dot{r}_o$ represents the deterministic \emph{velocity} disturbance due to the known road profile; $w$ represents the unknown road disturbance and $\alpha$ is a scaling factor; $c_s$ is the constant damping and $u$ is adjustable damper force; $k_s$ and $k_{us}$ are suspension and tire stiffness, respectively.

\begin{figure}[!t]
\centering
\includegraphics[width=2in]{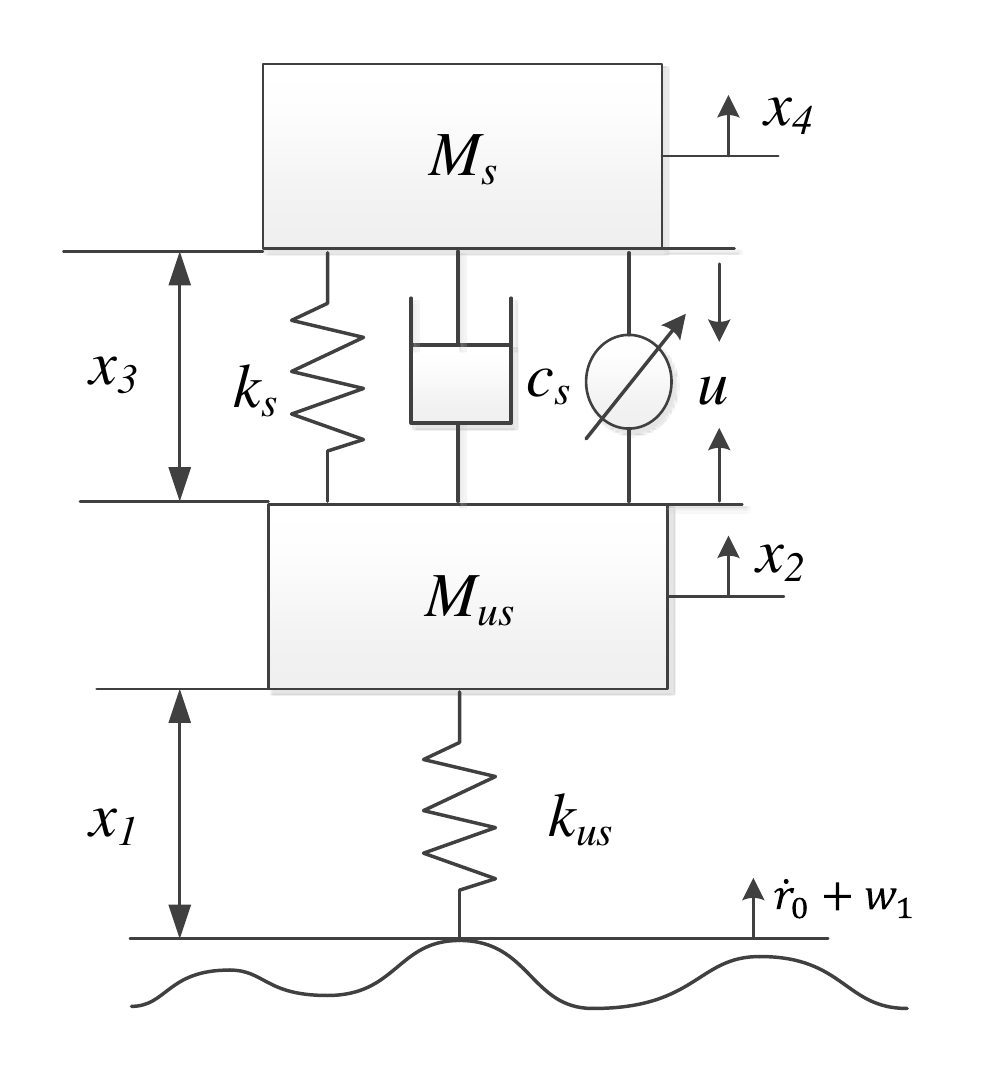}
\caption{Semi-active suspension dynamics.}
\label{fig:a}
\end{figure}

Defining $x=[x_1\quad x_2\quad x_3\quad x_4]^T$, the suspension system model can be written as
\begin{equation}\label{equ:new2}
\dot{x}=Ax+Bu+B_r\dot{r}_o+B_ww,
\end{equation}
where
\begin{equation}\label{AB}
A=
\begin{bmatrix}
0 &1 &0 &0\\[0.3em]
-\frac{k_{us}}{M_{us}} &-\frac{c_s}{M_{us}} &\frac{k_s}{M_{us}} &\frac{c_s}{M_{us}}\\[0.3em]
0 &-1 &0 &1\\[0.3em]
0 &\frac{c_s}{M_s} &-\frac{k_s}{M_s} &-\frac{c_s}{M_s}p
\end{bmatrix},
B=
\begin{bmatrix}
0\\
\frac{1}{M_{us}}\\[0.3em]
0\\
-\frac{1}{M_s}
\end{bmatrix},\\[0.3em]
\end{equation}
$$B_r=[-1\quad0\quad0\quad0]^{\text{T}},\quad B_w=[-\alpha\quad0\quad0\quad0]^\text{T}.$$
For vehicles equipped with semi-active suspension, measurements of vertical wheel velocity ($x_2$), suspension deflection ($x_3$) and body velocity ($x_4$) are typically available, while tire deflection is hard to measure. Let $y_0$ denote the measurements and $z$ denote the objective signal to be estimated, we augment (\ref{equ:new2}) as follows,
\begin{equation}\label{equ:new3}
\begin{aligned}
\dot{x}&=Ax+Bu+B_r\dot{r}_o+B_ww,\\
y_0&=[x_2\quad x_3\quad x_4]^{\text{T}}=C_0x+D_0w,\\
z&=x,
\end{aligned}
\end{equation}
where $C_0=[0_{3\times1}\quad I_3]$.

Figure~\ref{fig:diagrammain} illustrates the developed cloud-based vehicle software agent that has access to stored vehicle parameters ($M_{us}$,
$M_s$, $k_{us}$, $k_s$, $R$, $c_{s,i}$), receives vehicle state estimate, $\hat{x}$, vehicle longitudinal velocity, $v_{car}$, wheel speed, $\omega$, and GPS coordinates, and sends road profile information, $\dot{r}_o$ for use by on-board vehicle state estimator. The received road profile will be delayed in the wireless communication channel. Thus, we will have a delayed measurement of the road profile onboard, that is
\begin{equation}\label{r}
y_1(t)=\dot{r}_0(t-\tau(t))+D_1w(t),
\end{equation}
where $\tau(t)$ is the time-varying delay.
\begin{figure}[!t]
\centering
\includegraphics[width=1.5in]{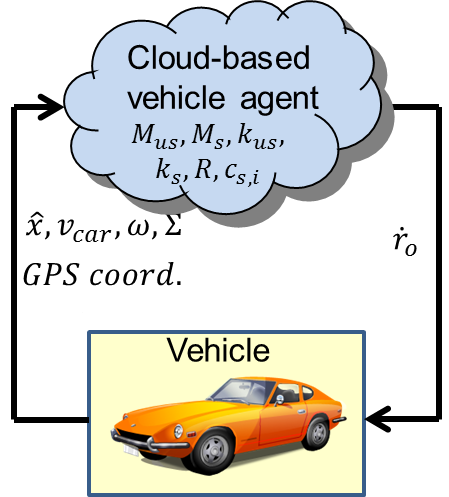}
\caption{V2C2V Suspension Control System.}\label{fig:diagrammain}
\end{figure}

\begin{assumption}
The time delay $\tau(t)$ is time varying and $\tau_m\le\tau(t)\le\tau_M$, where $\tau_m$ and $\tau_M$ are the lower and upper bound of the delay. In the sequel, we use $\tau$ to represent $\tau(t)$ when there is no confusion.
\end{assumption}

For the filter design, we assume that the road profile is modeled as $\ddot{r}_0=D_rw$, and we let $u(t)\equiv0$. Then combining (\ref{equ:new3}) and defining the augmented state as $x_a=[x^{\text{T}}\quad \dot{r}_0^{\text{T}}]^{\text{T}}$ and augmented measurement output $y_a=\big[y_0^\text{T}\quad y_1^\text{T}\big]^\text{T}$, we have
\begin{equation}\label{aug}
\begin{aligned}
\dot{x}_a(t)&=A_ax_a(t)+B_aw(t),\\
y_a(t)&=C_{a0}x_a(t)+C_{a1}x_a(t-\tau(t))+D_aw(t),\\
z(t)&=E_ax_a(t),
\end{aligned}
\end{equation}
where $A_a=\begin{bmatrix} A &B_r\\
0 &0\end{bmatrix}$,\quad $B_a=\begin{bmatrix} B_w\\
D_r\end{bmatrix}$, \quad $C_{a0}=\begin{bmatrix} C_0 &0\\
0 &0\end{bmatrix}$,\quad
$C_{a1}=\begin{bmatrix} 0 &0\\
0 &I\end{bmatrix}$,\quad $D_a=\begin{bmatrix} D_0 \\
D_1\end{bmatrix}$,\quad $E_a=[I\quad 0]$.
\\[0.3em]

A linear time invariant filter for system (\ref{aug}) has the following form:
\begin{equation}\label{filter}
\begin{aligned}
\dot{\hat{x}}(t)&=K_A\hat{x}(t)+K_By_a(t),\\
\hat{z}(t)&=K_C\hat{x}(t),\\
\hat{x}(0)&=x_a(0),\, \hat{x}(s)=0, \forall s\in[-\tau,0].
\end{aligned}
\end{equation}
Let $\eta(t)=\begin{bmatrix} x_a(t)\\\hat{x}(t)\end{bmatrix}$ denote the extended state. Then we have the following dynamics
\begin{equation}\label{extended}
\begin{aligned}
\dot{\eta}(t)&=\bar{A}\eta(t)+\bar{A}_dx_a(t-\tau)+\bar{B}w(t),\\
e(t)&\triangleq z(t)-\hat{z}(t)=\bar{C}\eta(t),
\end{aligned}
\end{equation}
where
$$\bar{A}=\begin{bmatrix}A_a &0\\K_BC_{a0} &K_A\end{bmatrix}, \quad \bar{A}_d=\begin{bmatrix}0 \\K_BC_{a1} \end{bmatrix},\quad
\bar{B}=\begin{bmatrix}B_a\\K_BD_a\end{bmatrix},\quad\bar{C}=[E_a\quad -K_C].$$ The desired $H_{\infty}$ filter problem addressed in this chapter can be formulated as follows: given system (\ref{aug}) and a prescribed level of noise attenuation $\gamma>0$, determine a linear filter in the form (\ref{filter}) such that the filtering error system is asymptotically stable and
\begin{equation}\label{gamma}
\underset{w\in L_2[0,\infty)}{\sup}\,\frac{\|e(t)\|_2^2}{\|w\|_2^2}<\gamma^2.
\end{equation}
Before ending this section, we present the following lemmas which will be used in the proofs of subsequent sections.
\\
\begin{lemma} \cite{LMI}: The linear matrix inequalities
$$S=\begin{bmatrix}S_{11}&S_{12}\\
S_{12}^\text{T}&S_{22}\end{bmatrix}<0$$
where $S_{11}=S_{11}^{\text{T}}$ and $S_{22}=S_{22}^{\text{T}}$ are equivalent to
$$S_{11}<0, \quad S_{22}-S_{12}^\text{T}S_{11}^{-1}S_{12}<0$$
or
$$S_{22}<0, \quad S_{11}-S_{12}S_{22}^{-1}S_{12}^\text{T}<0.$$
\end{lemma}
Lemma 1 is the well-know Schur complement lemma.
\begin{lemma} \cite{boukas}: Let $X$, $Y$ be real constant matrices of compatible dimensions. Then
$$X^{\text{T}}Y+Y^{\text{T}}X\le\epsilon X^{\text{T}}X+\frac{1}{\epsilon}Y^{\text{T}}Y$$
holds for any $\epsilon>0.$
\end{lemma}
\begin{lemma}\cite{like}
If $f$, $g$: $[a\quad b]\rightarrow \mathbb{R}^n$ are similarly ordered, that is,
\begin{equation}
(f(x)-f(y))^\text{T}(g(x)-g(y))\ge0, \quad \forall x,y\in[a,\,b],
\end{equation}
then,
\begin{equation}\label{Cheb}
\begin{aligned}
&\frac{1}{b-a}\int_a^bf(x)g(x)dx\ge\Big[\frac{1}{b-a}\int_a^bf(x)dx\Big]\Big[\frac{1}{b-a}\int_a^bg(x)dx\Big].
\end{aligned}
\end{equation}
\end{lemma}
Lemma 3 is a Chebyshev's inequality under similarly-ordered conditions.
\section{$H_\infty$ performance analysis}\label{sec:III}
In this section, sufficient conditions of existence of the $H_{\infty}$ filter are derived as LMIs.
\\
\begin{theorem} Let $K_A$, $K_B$, $K_C$ be given matrices and $\gamma$ be a given positive scalar. If there exist symmetric matrices $P>0$, $Q_1>0$, $Q_2>0$, satisfying
\begin{equation}\label{lemma3}
\begin{bmatrix}
\Psi &P\bar{B} &\bar{C}^{\text{T}} &\Gamma\\
* &-\gamma^2I+\tau_M Q_2 &0 &0\\
* &* &-I &0\\
* &* &* &-\frac{1}{\tau_M}diag\{Q_1,Q_2\}
\end{bmatrix}<0
\end{equation}
where $$\Psi=P[\bar{A}+\bar{A}_dI_0]+[\bar{A}+\bar{A}_dI_0]^{\text{T}}P+\tau_M I_0^{\text{T}}Q_1I_0,\quad
\Gamma=[P\bar{A}_dA_a \quad P\bar{A}_dB_a],\quad I_0=[I\quad 0],$$
then the error system (\ref{extended}) is asymptotically stable and satisfies (\ref{gamma}).
\end{theorem}
\begin{proof} Note:
\begin{equation}\label{equ:2}
\begin{aligned}
x_a(t-\tau)&=x_a(t)-\int_{t-\tau}^t\,\dot{x}_a(s)ds\\
&=x_a(t)-\int_{t-\tau}^t\,[A_ax_a(s)+B_aw(s)]ds.
\end{aligned}
\end{equation}
Substituting (\ref{equ:2}) into (\ref{extended}) leads to
\begin{equation}\label{sub}
\begin{aligned}
\dot{\eta}(t)&=[\bar{A}+\bar{A}_dI_0]\eta(t)-\bar{A}_d\int_{t-\tau}^t[A_ax_a(s)+B_aw(s)]ds+\bar{B}w(t).
\end{aligned}
\end{equation}
We next show the stability of the error system (\ref{extended}). Let us consider the following Lyapunov functional candidate
\begin{equation}\label{lya}
V(\eta(t))=V_0(\eta(t))+V_1(\eta(t))+V_2(\eta(t)),
\end{equation}
where
$$V_0(\eta(t))=\eta^{\text{T}}(t)P\eta(t),$$
$$
V_1(\eta(t))=\int_{-\tau_M}^0\int_{t+\theta}^tx_a^{\text{T}}(s)Q_1x_a(s)dsd\theta,$$
$$
V_2(\eta(t))=\int_{-\tau_M}^0\int_{t+\theta}^tw^{\text{T}}(s)Q_2w(s)dsd\theta.$$

In view of (\ref{sub}), we have
\begin{equation}\label{der}
\begin{aligned}
\dot{V}_0(\eta(t))=&\eta^{\text{T}}(t)\big[P[\bar{A}+\bar{A}_dI_0]+[\bar{A}+\bar{A}_dI_0]^{\text{T}}P\big]\eta(t)\\
&-2\eta^{\text{T}}(t)P\bar{A}_d\big[A_a\int_{t-\tau}^tx_a(s)ds+B_a\int_{t-\tau}^tw(s)ds\big]+2\eta^{\text{T}}(t)P\bar{B}w(t).
\end{aligned}
\end{equation}
Let $X^\text{T}=-\eta^\text{T}(t)P\bar{A}_dA_aQ_1^{-\frac{1}{2}}$, $\epsilon=\tau$, $ Y=Q_1^{\frac{1}{2}}\int_{t-\tau}^tx_a(s)ds$. Using Lemma 2, we obtain
\begin{equation}\label{equ:3}
\begin{aligned}
-2\eta^{\text{T}}(t)P\bar{A}_dA_a\int_{t-\tau}^tx_a(s)ds&\leq\tau\eta^{\text{T}}(t)[P\bar{A}_dA_aQ_1^{-1}A_a^\text{T}\bar{A}_d^\text{T}P]\eta(t)\\
&+\frac{1}{\tau}\int_{t-\tau}^tx_a^{T}dsQ_1\int_{t-\tau}^tx_a(s)ds.
\end{aligned}
\end{equation}
Using Lemma 3, we have
\begin{equation}\label{equ:che1}
\int_{t-\tau}^tx_a^{T}(s)dsQ_1\int_{t-\tau}^tx_a(s)ds\leq \tau \int_{t-\tau}^tx_a^\text{T}(s)Q_1x_a(s)ds.
\end{equation}
From (\ref{equ:3}) and (\ref{equ:che1}), it follows that
\begin{equation}\label{equ:che}
\begin{aligned}
-2\eta^{\text{T}}(t)P\bar{A}_dA_a\int_{t-\tau}^tx_a(s)ds&\leq\tau\eta^{\text{T}}(t)[P\bar{A}_dA_aQ_1^{-1}A_a^\text{T}\bar{A}_d^\text{T}P]\eta(t)\\
&+\int_{t-\tau}^tx_a^{T}Q_1x_a(s)ds.
\end{aligned}
\end{equation}
Similarly,
\begin{equation}\label{equ:4}
\begin{aligned}
-2\eta^\text{T}(t)P\bar{A}_dB_a\int_{t-\tau}^tw(s)ds&\leq\tau\eta^\text{T}(t)P\bar{A}_dB_aQ_2^{-1}B_a^\text{T}\bar{A}_d^\text{T}P\eta(t)\\
&+\int_{t-\tau}^tw^\text{T}(s)Q_2w(s)ds.
\end{aligned}
\end{equation}
Combining (\ref{der})-(\ref{equ:4}) yields
\begin{equation}\label{equ:5}
\begin{aligned}
\dot{V}_0(\eta(t))&\leq\eta^{\text{T}}(t)\big[P[\bar{A}+\bar{A}_dI_0]+[\bar{A}+\bar{A}_dI_0]^{\text{T}}P\big]\eta(t)\\
&+\tau\eta^{\text{T}}(t)[P\bar{A}_dA_aQ_1^{-1}A_a^\text{T}\bar{A}_d^\text{T}P]\eta(t)\\
&+\int_{t-\tau}^tx_a^{T}(s)Q_1x_a(s)ds\\
&+\tau\eta^\text{T}(t)P\bar{A}_dB_aQ_2^{-1}B_a^\text{T}\bar{A}_d^\text{T}P\eta(t)\\
&+\int_{t-\tau}^tw^\text{T}(s)Q_2w(s)ds+2\eta^{\text{T}}(t)P\bar{B}w(t).
\end{aligned}
\end{equation}
Simple computations yield
\begin{equation}\label{equ:6}
\begin{aligned}
&\dot{V}_1(\eta(t))=-\int_{t-\tau_M}^tx_a^\text{T}(s)Q_1x_a(s)ds+\tau_M x_a^\text{T}(t)Q_1x_a(t)
\end{aligned}
\end{equation}
Also,
\begin{equation}\label{equ:66}
\dot{V}_2(\eta(t))=\tau_M w^\text{T}(t)Q_2w(t)-\int_{t-\tau_M}^tw^\text{T}(s)Q_2w(s)ds.
\end{equation}
Combining (\ref{equ:5}), (\ref{equ:6}) and (\ref{equ:66}), we obtain
\begin{equation}\label{equ:77}
\begin{aligned}
\dot{V}(\eta(t))&\le \eta^{\text{T}}(t)\big[P[\bar{A}+\bar{A}_dI_0]+[\bar{A}+\bar{A}_dI_0]^{\text{T}}P\big]\eta(t)\\
&\quad+\tau_M\eta^{\text{T}}(t)[P\bar{A}_dA_aQ_1^{-1}A_a^\text{T}\bar{A}_d^\text{T}P]\eta(t)\\
&\quad+\tau_M\eta^\text{T}(t)P\bar{A}_dB_aQ_2^{-1}B_a^\text{T}\bar{A}_d^\text{T}P\eta(t)\\
&\quad+\tau_M x_a^\text{T}(t)Q_1x_a(t)+\tau_M w^\text{T}(t)Q_2w(t)+2\eta^{\text{T}}(t)P\bar{B}w(t)\\
&=[\eta^\text{T}(t)\, w^\text{T}(t)]\Phi\begin{bmatrix}\eta(t)\\w(t)\end{bmatrix},
\end{aligned}
\end{equation}
where $$\Phi=\begin{bmatrix}\Psi &P\bar{B}\\
\bar{B}^\text{T}P &\tau_M Q_2\end{bmatrix}+\begin{bmatrix}\tau_M\Gamma diag\{Q_1^{-1},\,Q_2^{-1}\}\Gamma^\text{T} &0\\
0 &0
\end{bmatrix}.$$
Using Lemma 1, it is straightforward to check that (\ref{lemma3}) implies $\Phi<0$, which concludes the proof of stability.\\
Now let's define an $H_{\infty}$ performance $J_{T}$ as follows:
$$J_{T}=\int_0^T[e^\text{T}(t)e(t)-\gamma^2w^\text{T}(t)w(t)]dt.$$
Then we have
\begin{equation}
\begin{aligned}
J_T&=\int_0^T[e^\text{T}(t)e(t)-\gamma^2w^\text{T}(t)w(t)]dt\\
&=\int_0^T[e^\text{T}(t)e(t)-\gamma^2w^\text{T}(t)w(t)+\dot{V}(\eta(t))-\dot{V}(\eta(t))]dt\\
&=\int_0^T[e^\text{T}(t)e(t)-\gamma^2w^\text{T}(t)w(t)+\dot{V}(\eta(t))]dt-V(\eta(T))\\
&\le \int_0^T [\eta^\text{T}(t)\, w^\text{T}(t)]\Theta\begin{bmatrix}\eta(t)\\w(t)\end{bmatrix}dt,
\end{aligned}
\end{equation}
where
\begin{equation}
\begin{aligned}
\Theta&=\begin{bmatrix}\Psi+\bar{C}^\text{T}\bar{C} &P\bar{B}\\
\bar{B}^\text{T}P &-\gamma^2I+\tau_M Q_2\end{bmatrix}
+\begin{bmatrix}\tau_M\Gamma diag\{Q_1^{-1},\,Q_2^{-1}\}\Gamma^\text{T} &0\\
0 &0
\end{bmatrix}.
\end{aligned}
\end{equation}
Using Schur complement as in Lemma 1, it is straight to check that (\ref{lemma3}) implies $\Theta<0$ and consequently
$$J_T<0\quad \forall T>0,$$
and it follows that
$$J_{\infty}=\int_0^\infty[e^\text{T}(t)e(t)-\gamma^2w^\text{T}(t)w(t)]dt<0,$$
that is $\|e\|_2\le\gamma\|w\|_2$. This concludes the proof of Theorem 1.
\end{proof}

\section{Filter design}\label{sec:IV}
In this section, we present the design of filter gains $K_A$, $K_B$ and $K_C$ in (\ref{filter}).\\
We note that using Schur complement, it can be shown that (\ref{lemma3}) is equivalent to
\begin{equation}\label{equ:7}
\begin{bmatrix}
\Psi_1 &P\bar{B} &\bar{C}^{\text{T}} &\Gamma_1\\
* &-\gamma^2I+\tau_M Q_2 &0 &0\\
* &* &-I &0\\
* &* &* &-1/\tau_M\mathcal{Q}
\end{bmatrix}<0
\end{equation}
where $$\Psi_1=P[\bar{A}+\bar{A}_dI_0]+[\bar{A}+\bar{A}_dI_0]^{\text{T}}P,\quad\Gamma_1=[P\bar{A}_dA_a\; P\bar{A}_dB_a\;I_0^\text{T}],$$
$$\mathcal{Q}=diag\{Q_1,Q_2,Q_1^{-1}\}.$$
Assume that $P$ and $P^{-1}$ can be decomposed as follows,
\begin{equation}\label{decompose}
P=\begin{bmatrix}Y &N\\
N^\text{T} &W_1\end{bmatrix},\quad P^{-1}=\begin{bmatrix}X &M\\M^\text{T} &W_2\end{bmatrix}.\end{equation}
Then $PP^{-1}=I$ implies
\begin{equation}\label{equ:8}
\begin{cases}
&YX+NM^\text{T}=I,\\
&YM+NW_2=0,\\
&N^\text{T}X+W_1M^\text{T}=0.
\end{cases}
\end{equation}
Define
\begin{equation}\label{phi}
\Phi_1=\begin{bmatrix}X &I\\
M^\text{T} &0\end{bmatrix},\quad \Phi_2=\begin{bmatrix}I &Y\\0 &N^\text{T}\end{bmatrix}.
\end{equation}
It can be shown that
\begin{equation}\label{P}
P\Phi_1=\Phi_2
\end{equation}
and
\begin{equation}\label{haha}
\Phi_1^\text{T}P\Phi_1=\begin{bmatrix}X &I\\I &Y\end{bmatrix}>0.
\end{equation}
Define the following matrices $\mathcal{A}$, $\mathcal{B}$, $\mathcal{C}$:
\begin{equation}\label{abc}
\begin{cases}
&\mathcal{A}=YA_aX+NK_B(C_{a0}+C_{a1})X+NK_AM^\text{T},\\
&\mathcal{B}=NK_B,\\
&\mathcal{C}=K_CM^\text{T}.
\end{cases}
\end{equation}
With direct matrix calculation, pre- and post- multiplying (\ref{equ:7}) by $diag\{\Phi_1^\text{T},I,I,I\}$ and $diag\{\Phi_1,I,I,I\}$, respectively, yields
\begin{equation}\label{equ:10}
\begin{bmatrix}
T_1 &T_2 &T_3 &T_4\\
* &-\gamma^2+\tau_M Q_2 &0 &0\\
* &* &-I &0\\
* &* &* &-1/\tau_M\mathcal{Q}
\end{bmatrix}<0,
\end{equation}
where $$T_1=\begin{bmatrix}A_aX+XA_a^\text{T} &A_a+\mathcal{A}^\text{T}\\
A_a^\text{T}+\mathcal{A} &\begin{pmatrix}&YA_a+A_a^\text{T}Y+\mathcal{B}(C_{a0}+C_{a1})\\&+(C_{a0}+C_{a1})^\text{T}\mathcal{B}^\text{T}\end{pmatrix}\end{bmatrix},$$
$$T_2=\begin{bmatrix}B_a\\YB_a+\mathcal{B}D_a\end{bmatrix},\quad T_3=\begin{bmatrix}XE_a^\text{T}-\mathcal{C}^\text{T}\\E_a^\text{T}\end{bmatrix},$$
$$T_4=\begin{bmatrix}0 &0 &X \\
\mathcal{B}C_{a1}A_a &\mathcal{B}C_{a1}B_a &I \end{bmatrix}.$$
Thus, if there exist symmetric matrices $X>0$, $Y>0$ and matrices $\mathcal{A}$, $\mathcal{B},$ $\mathcal{C}$ that satisfy (\ref{equ:10}), then the $H_{\infty}$ criterion is satisfied and filter gains can be obtained by solving (\ref{abc}). We next give the following Lemma which can be used to design the filter
\begin{theorem}
If there exist symmetric matrices $X>0$, $Y>0$ and matrices $\mathcal{A}$, $\mathcal{B},$ $\mathcal{C}$ that satisfy (\ref{equ:10}) and (\ref{haha}) with a given positive constant $\gamma$, then there exist matrices $K_A$, $K_B$, $K_C$ such that the error system (\ref{extended}) is stable and (\ref{gamma}) is satisfied. The filter gains are
$$K_A=N^{-1}[\mathcal{A}-YA_aX-\mathcal{B}(C_{a0}+C_{a1})X](M^\text{T})^{-1},$$
$$K_B=N^{-1}\mathcal{B},\quad K_C=\mathcal{C}(M^\text{T})^{-1}.$$
\end{theorem}
\begin{proof}
(\ref{haha}) guarantees that there is a positive definite matrix $P$ that can be decomposed as in (\ref{decompose}). Then it is easy to check that defined $K_A$, $K_B$, $K_C$ and $P$ satisfies (\ref{lemma3}).\\
\emph{Remark:} To implement the algorithm, first solve LMIs (\ref{haha}) and (\ref{equ:10}) to get $X$, $Y$, $\mathcal{A}$, $\mathcal{B}$ and $\mathcal{C}$. Then using $MN^\text{T}=I-XY$ as in (\ref{equ:8}), $M$ and $N$ can be obtained by singular value decomposition. Then define $\Phi_1$ and $\Phi_2$ as in (\ref{phi}), $P$ can be solved using (\ref{P}). Then the filter gains can be computed as in Theorem 2.
\end{proof}

\section{Simulations}\label{sec:V}
In this section, we present numerical simulations to illustrate the effectiveness of our designed filter. The simulation parameters are illustrated in Table~\ref{tab1}.
\begin{table}[h]
\caption{Simulation parameters}
\label{tab1}
\centering
\begin{tabular}{|c |c |c |c|c|c| }
\hline
\hline
$m_s$ &$m_{us}$ &$k_s$ &$k_{us}$ &$c_s$ &$\alpha$ \\
\hline
 290 $kg$ &60 $kg$ &16800 $N/m$ &19000 $N/m$ &200 $N\cdot s/m$ &0.1\\
\hline
\end{tabular}
\end{table}
For simulations, a road segment over a 10 sec horizon is modeled as follows,

$$
\dot{r}_o(t)=
\begin{cases}
0.15\cdot\sin{\pi (t-1)} &1s \leq t \leq 3s, \\
0.2\cdot\sin{\pi/2 t} &4s \leq t \leq 8s, \\
0 & \rm{otherwise}.
\end{cases}
$$
See Figure \ref{fig:e}.
\begin{figure}[h]
\centering
\includegraphics[scale=0.3]{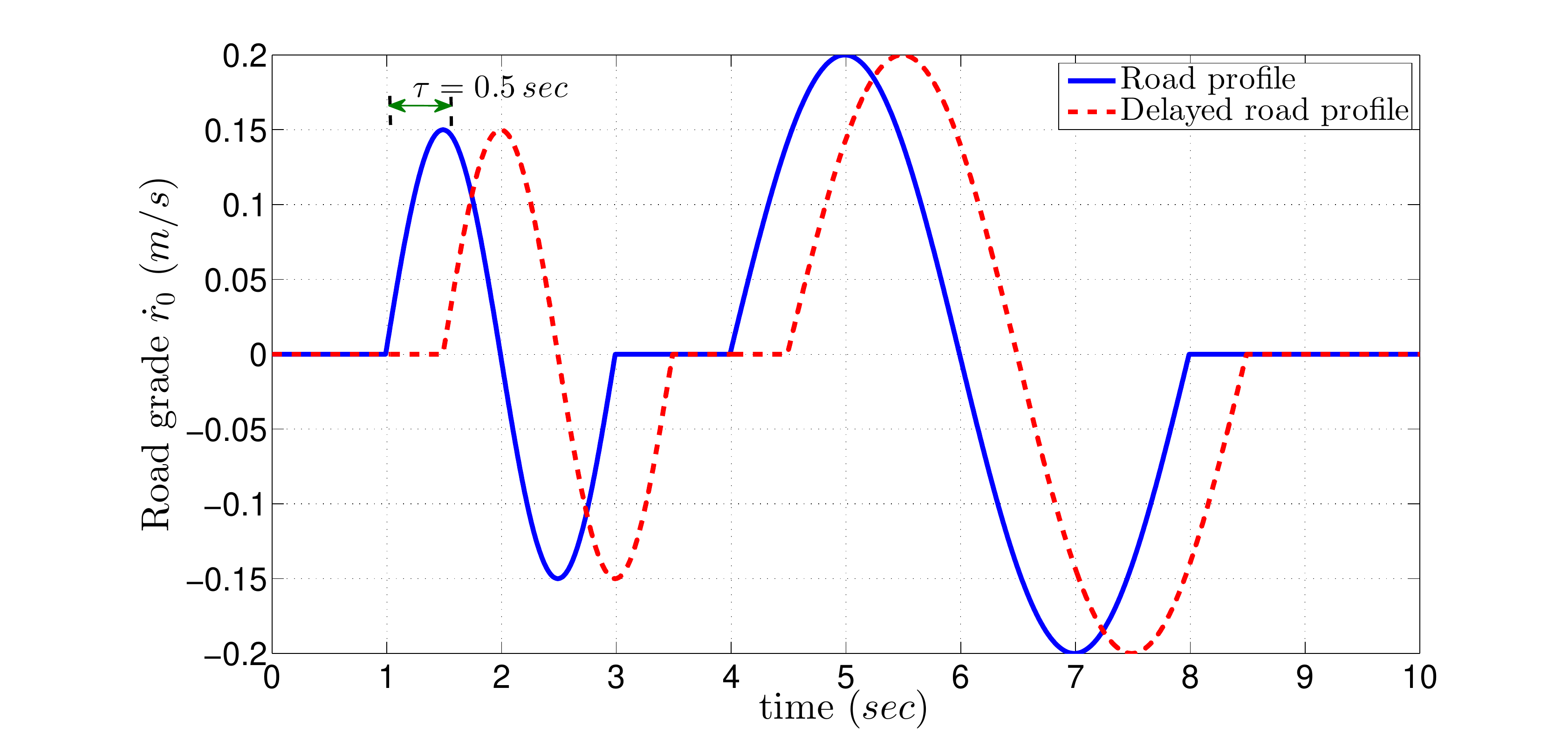}
\caption{Road grade profile ($\dot{r}_0$).}
\label{fig:e}
\end{figure}
Let $\gamma=0.5$ and $\tau_M=0.5$ $sec$, we aim at a filter in form of (\ref{filter}) such that (\ref{gamma}) is satisfied. Applying Theorem 2 with Matlab LMI toolbox, filter gains $K_A$, $K_B$ and $K_C$ in (\ref{filter}) are obtained.
%
%

 With the obtained filter, estimates of $x_3$ and $x_4$ are shown in Fig. \ref{fig:a3} and Fig. \ref{fig:a4}, respectively. Plots of estimates of $x_1$ and $x_2$ are not shown due to space limit. The $H_{\infty}$ filter performance is compared with a Kalman filter assuming no knowledge of the road profile information. It can be seen that by using road profile from the cloud with small delays (e.g., 0.2 sec), our designed $H_{\infty}$ filtering has better performances than the traditional Kalman filter. However, with larger time delays, the estimation performance can be worse than the Kalman filter. Note that from practical standpoint the delays in V2C2V system can be reduced if sufficient communication bandwidth and on-board memory storage is available so that the road profile information can be transmitted in advance.
\begin{figure}[h]
\centering
\includegraphics[width=5in]{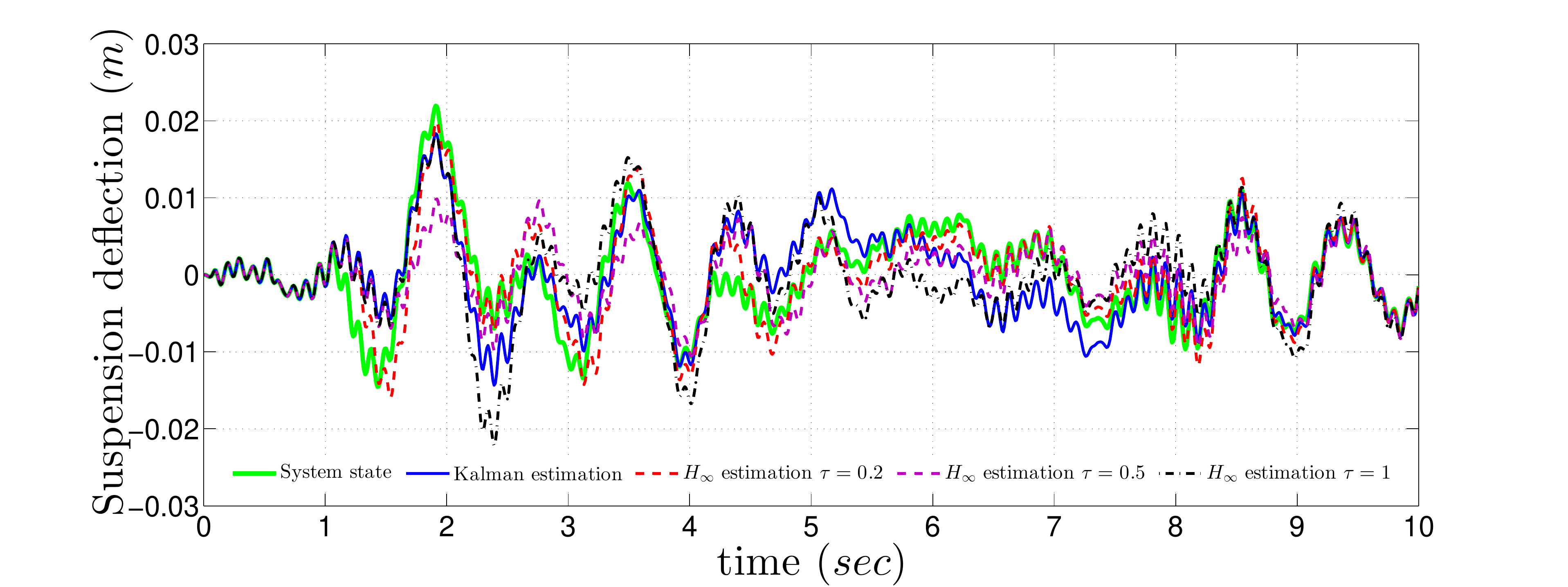}
\caption{Suspension deflection $x_3$}
\label{fig:a3}
\end{figure}

\begin{figure}[h]
\centering
\includegraphics[width=5in]{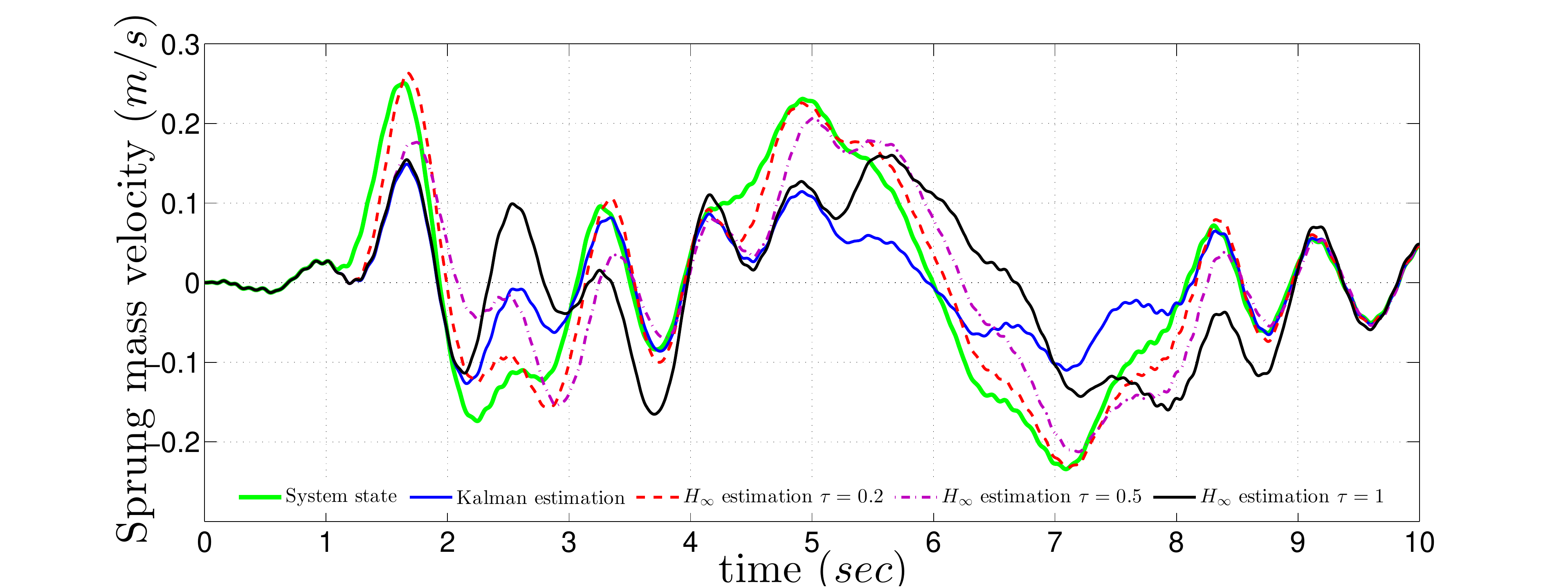}
\caption{Sprung mass velocity $x_4$}
\label{fig:a4}
\end{figure}
\section{Conclusions}\label{sec:VI}
In this chapter we studied an $H_{\infty}$ filtering problem for cloud-aided semi-active suspension where road profile information was sent from the cloud to the vehicle to compliment on-board measurements. We have studied this problem under the assumption that there are delays in transmitting the information from the cloud to the vehicle. Sufficient conditions of existence of the $H_{\infty}$ filter were given in terms of linear matrix inequalities. The explicit expressions of the filter parameters were derived.  A numerical example illustrated that the proposed filter framework has a potential for performance than traditional Kalman filter when the communication delay is not very large.

\bibliographystyle{ieeetr}
\bibliography{semiactive}
\end{document}